\documentclass[12pt]{l4dc2021} 

\usepackage{caption}
\newcommand{\uppercaseGreek}[1]{%
    \begingroup
    \ucmathlist\MakeUppercase{#1}%
    \endgroup
}

\newcommand{\ucmathlist}{%
    \def\alpha{\mathrm{A}}%
    \def\beta{\mathrm{B}}%
    \let\gamma=\Gamma
    \let\delta=\Delta
    \def\epsilon{\mathrm{E}}%
    \def\varepsilon{\mathrm{E}}%
    \def\zeta{\mathrm{Z}}%
    \def\eta{\mathrm{H}}%
    \let\theta=\Theta
    \let\vartheta=\Theta
    \def\iota{\mathrm{I}}%
    \def\kappa{\mathrm{K}}%
    \let\lambda=\Lambda
    \def\mu{\mathrm{M}}%
    \def\nu{\mathrm{N}}%
    \let\xi=\Xi
    \let\pi=\Pi
    \let\varpi=\Pi
    \def\rho{\mathrm{P}}%
    \def\varrho{\mathrm{P}}%
    \let\sigma=\Sigma
    \def\tau{\mathrm{T}}%
    \let\upsilon=\Upsilon
    \let\phi=\Phi
    \let\varphi=\Phi
    \def\chi{\mathrm{X}}%
    \let\psi=\Psi
    \let\omega=\Omega
}

\makeatletter
\def\renewtheorem#1{%
    \expandafter\let\csname#1\endcsname\relax
    \expandafter\let\csname c@#1\endcsname\relax
    \gdef\renewtheorem@envname{#1}
    \renewtheorem@secpar
}
\def\renewtheorem@secpar{\@ifnextchar[{\renewtheorem@numberedlike}{\renewtheorem@nonumberedlike}}
\def\renewtheorem@numberedlike[#1]#2{\newtheorem{\renewtheorem@envname}[#1]{#2}}
\def\renewtheorem@nonumberedlike#1{  
    \def\renewtheorem@caption{#1}
    \edef\renewtheorem@nowithin{\noexpand\newtheorem{\renewtheorem@envname}{\renewtheorem@caption}}
    \renewtheorem@thirdpar
}
\def\renewtheorem@thirdpar{\@ifnextchar[{\renewtheorem@within}{\renewtheorem@nowithin}}
\def\renewtheorem@within[#1]{\renewtheorem@nowithin[#1]}
\makeatother
\input{auxFiles/berkeleyColors.sty}
\input{auxFiles/symbolDef.sty}

\title[Graph neural networks for distributed linear-quadratic control]{Graph Neural Networks for Distributed Linear-Quadratic Control}
\usepackage{times}

\author{\Name{Fernando Gama\nametag{\thanks{This work was supported by grants from ONR, NSF and AFOSR. }}} \Email{fgama@berkeley.edu}\and
    \Name{Somayeh Sojoudi} \Email{sojoudi@berkeley.edu}\\
    \addr Electrical Engineering and Computer Sciences Dept., University of California, Berkeley, CA 94709, USA}

\begin{document}

\maketitle

\begin{abstract}%
The linear-quadratic controller is one of the fundamental problems in control theory. The optimal solution is a linear controller that requires access to the state of the entire system at any given time. When considering a network system, this renders the optimal controller a \emph{centralized} one. The interconnected nature of a network system often demands a \emph{distributed} controller, where different components of the system are controlled based only on local information. Unlike the classical centralized case, obtaining the optimal distributed controller is usually an intractable problem. Thus, we adopt a graph neural network (GNN) as a parametrization of distributed controllers. GNNs are naturally local and have distributed architectures, making them well suited for learning nonlinear distributed controllers. By casting the linear-quadratic problem as a self-supervised learning problem, we are able to find the best GNN-based distributed controller. We also derive sufficient conditions for the resulting closed-loop system to be stable. We run extensive simulations to study the performance of GNN-based distributed controllers and showcase that they are a computationally efficient parametrization with scalability and transferability capabilities.%
\end{abstract}



\section{Introduction} \label{sec:intro}



Undoubtedly, linear dynamical systems are the cornerstone of countless information processing algorithms in a wide array of areas, including physics, mathematics, engineering and economics \citep{Kailath80-LinearSystems, Rugh96-LinearSystems}. Therefore, the ability to optimally control these systems is of paramount importance \citep{Kwakernaak72-LinearOptimal, Bertsekas05-DynamicProgramming}. Of particular interest is the case when linear systems are coupled with a quadratic cost, giving rise to the well-studied linear-quadratic control problem \citep{AndersonMoore89-LQR, Makila87-ComputationalLQ, Dean20-SampleLQR}. As it happens, the optimal linear-quadratic controller is \emph{linear} and acts on the knowledge of the system state at a given time to produce the optimal control action for that time instant.

A special class of dynamical systems that has gained widespread attention are network systems. These systems are comprised of a set of interconnected components, capable of exchanging information and equipped with the ability to autonomously decide on an action to take. The objective is to coordinate the individual actions of the components so that they are conducive to the accomplishment of some global task \citep{Bamieh02-DistributedControl, Nayyar13-DecentralizedStochastic, Fattahi19-SingleSample}.

Linear network systems that seek to minimize a global quadratic cost can be easily controlled if we allow for a \emph{centralized} approach. That is, if we assume that all components have instantaneous access to the state of all other components, they can readily compute the optimal control action. Furthermore, such an optimal centralized controller is linear. Computing optimal centralized controllers, however, face limitations in terms of scalability and implementation.

The interconnected nature of network systems naturally imposes a distributed data structure. While this structure may affect the linear dynamics of the system, or even the objective quadratic cost, we are predominantly interested in leveraging the data structure to obtain \emph{distributed} controllers. These are control actions that depend only on local information provided by components that share a connection and that can be computed separately by each component.

Imposing a distributed constraint on the linear-quadratic control problem renders it intractable in the most general case \citep{Witsenhausen68-Counterexample, Tsitsiklis84-DecentralizedComplexity}. While there is a large class of distributed control problems that admit a convex formulation \citep{Rotkowitz06-DecentralizedConvex}, many of them lead to complex solutions that do not scale with the size of the network \citep{Tanaka14-TriangularLQG, Fattahi19-TransformCentralized}. An alternative approach is to adopt a linear parametrization of the controller and find a surrogate of the original problem that admits a scalable solution. The resulting controller is thus a sub-optimal linear distributed controller. Suboptimality guarantees are often obtained, and stability and robustness analysis of the resulting controller are provided \citep{Fazelnia17-LowerBoundLQR, Matni19-SystemLevelApproach, Fattahi19-LQR}

However, even in the context of linear network system with a quadratic cost, the optimal distributed controller may not be linear \citep{Witsenhausen68-Counterexample}. In this paper, we thus adopt a nonlinear parametrization of the controller. More specifically, we focus on the use of graph neural networks (GNNs) \citep{Bronstein17-GeometricDeepLearning, Gama20-GNNs}. GNNs consist of a cascade of blocks (commonly known as layers) each of which applies a bank of graph filters followed by a pointwise nonlinearity \citep{Bruna14-DeepSpectralNetworks, Defferrard17-ChebNets, Gama19-Archit}. GNNs exhibit several desirable properties in the context of distributed control. Most importantly, they are naturally local and distributed, meaning that by adopting a GNN as a mapping between states and actions, a distributed controller is automatically obtained. Furthermore, they are permutation equivariant and Lipschitz continuous to changes in the network \citep{Gama20-Stability}. These two properties allow them to scale up and transfer to unknown networks \citep{Ruiz20-Transferability}.

Distributed controllers leveraging neural network techniques can be found in \citep{Capella03-DistributedNN, Huang05-LargeScaleDecentralized, Srinivasan06-ContinuousOnlineLearning, Chen13-DecentralizedPID, Liu15-ContinuousTimeUnknown, Yang17-FormationControlRBF, Tolstaya19-Flocking, Li20-Planning}. These controllers typically use a distinct multi-layer perceptron (MLP) to parametrize the controller at each component \citep{Capella03-DistributedNN, Huang05-LargeScaleDecentralized, Srinivasan06-ContinuousOnlineLearning, Chen13-DecentralizedPID, Liu15-ContinuousTimeUnknown, Yang17-FormationControlRBF}, while GNNs are leveraged in \citep{Tolstaya19-Flocking, Li20-Planning} in the context of robotics problems.

The main contributions of this work are: (i) the use of graph neural networks (Section~\ref{sec:GNN}) as a practically useful nonlinear parametrization of the distributed controller for a linear-quadratic problem (Section~\ref{sec:LQR}), (ii) casting the linear-quadratic problem as a self-supervised learning problem that can be solved by traditional machine learning techniques [cf. \eqref{eq:decentralizedERM}], (iii) a sufficient condition for the resulting GNN-based controller to stabilize the system (Proposition~\ref{prop:stability}), and (iv) a comprehensive numerical simulation investigating the performance of GNN-based distributed controllers and its dependence on design hyperparameters, as well as its scalability and transferability (Section~\ref{sec:sims}). 




\section{Distributed Linear-Quadratic Controllers} \label{sec:LQR}



Consider a system of $N$ components, each one described at time $t \in \{0,1,2,\ldots,T\}$ by a state vector $\vcx_{i}(t) \in \fdR^{F}$ for $i=1,\ldots,N$. These components are also equipped with the ability to take an action $\vcu_{i}(t) \in \fdR^{G}$ that can influence future values of its own state, as well as the states of other components in the system, as determined by some given dynamic model . The state and the control actions can be compactly described by matrices $\mtX(t) \in \fdR^{N \times F}$, $\mtU(t) \in \fdR^{N \times G}$, respectively, with
\begin{equation} \label{eq:stateMatrix}
    \mtX(t) = 
        \begin{bmatrix} 
            \vcx_{1}(t)^{\Tr} \\ 
            \vdots            \\
            \vcx_{N}(t)^{\Tr}
        \end{bmatrix}
\end{equation}
and analogously for $\mtU(t)$, i.e. the $i^{\text{th}}$ row of $\mtU(t)$ is the action $\vcu_{i}(t) \in \fdR^{G}$ taken by component $i$ at time $t$. We are particularly interested in linear system dynamics, modeled as \citep[Ch. 6]{Kailath80-LinearSystems}
\begin{equation} \label{eq:linearSystem}
    \mtX(t+1) = \mtA \mtX(t) + \mtB \mtU(t) \mtbB
\end{equation}
with the given matrix $\mtA \in \fdR^{N \times N}$ known as the system matrix, and the matrices $\mtB \in \fdR^{N \times N}$ and $\mtbB \in \fdR^{G \times F}$ known as the control matrices. We note that $F$ is the dimension of the state vector at each component and $G$ is the dimension of the control action executed by each component. In this context, the matrix $\mtbB$ acts as a linear map between the $G$ values of each individual control action $\vcu_{i}(t)$ and the $F$ values of each individual state $\vcx_{i}(t)$.

A linear-quadratic regulator (LQR) is a controller that minimizes a quadratic cost on the states and the control actions \citep[Sec. 4.1]{Bertsekas05-DynamicProgramming}:
\begin{equation} \label{eq:quadraticCost}
    \fnJ \Big( \{\mtX(t)\}_{t=0}^{T}, \{\mtU(t)\}_{t=0}^{T-1} ;
                 \mtQ, \mtR \Big)
     = \big\| \mtQ^{1/2} \mtX(T) \big\|^{2} 
        + \sum_{t=0}^{T-1} \Big( 
              \big\| \mtQ^{1/2} \mtX(t) \big\|^{2} 
            + \big\| \mtR^{1/2} \mtU(t) \big\|^{2} 
          \Big)
\end{equation}
for some given positive semidefinite matrices $\mtQ, \mtR \in \fdR^{N \times N}$ and for some given matrix norm $\|\cdot\|$. We note that if we set $F=G=1$, then the state and control actions become vectors $\vcx(t),\vcu(t) \in \fdR^{N}$, respectively, and by choosing the Euclidean norm, the quadratic cost \eqref{eq:quadraticCost} results in the well-known cost of the traditional LQR problem \citep{AndersonMoore89-LQR}. In this case, the optimal control actions that solve the LQR problem can be computed directly by means of a linear map of the state value, i.e. $\vcu^{\opt}(t) = \mtK_{t}^{\opt} \vcx(t)$, with $\mtK_{t}^{\opt} \in \fdR^{N \times N}$ having a closed-form solution in terms of $\mtA, \mtB, \mtQ, \mtR$ \citep[Ch. 4]{Bertsekas05-DynamicProgramming}. In summary, the objective of the LQR problem is to find a sequence of control actions $\{\mtU(t)\}$ such that \eqref{eq:quadraticCost} is minimized, subject to the dynamics in \eqref{eq:linearSystem}.

In what follows, we focus on the case where the system is distributed in nature. This means that the components of the system can only interact with those other components to which they have a direct connection. To describe this connectivity pattern, we model the system as a graph $\stG = \{\stV, \stE\}$, where $\stV = \{\lmv_{1},\ldots,\lmv_{N}\}$ is the set of nodes with node $\lmv_{i}$ representing the $i^{\text{th}}$ component in the system and where $\stE \subseteq \stV \times \stV$ is the set of edges with $(\lmv_{i},\lmv_{j}) \in \stE$ if and only if the $i^{\text{th}}$ component is connected to the $j^{\text{th}}$ one. The distributed nature of the system typically has an impact on the structure of the system matrix $\mtA$ and the control matrix $\mtB$ [cf. \eqref{eq:linearSystem}], usually respecting the topology of the graph (i.e. the sparsity of the adjacency matrix). More generally, we assume that the matrices $\mtA$ and $\mtB$ share the sparsity of some $k$-hop shortest path between neighbors, i.e. that the $(i,j)$ entry of the matrix may be nonzero if and only if there is a $k$-hop path between $\lmv_{i}$ and $\lmv_{j}$ for some (unknown) value of $k$. We note that this may lead to non-sparse matrices, but that still exhibit a strong structure related to the underlying connectivity of the system. Examples include, both discrete-time \citep{Gama19-GLLN} and continuous-time \citep{Olfati07-Consensus} diffusion models, as well as heat processes \citep{Thanou17-Heat}, among others \citep{Gama19-Control}.

Another fundamental aspect where the distributed nature of the system has a key impact is on the controller. More specifically, we concentrate on finding controllers that minimize the quadratic cost \eqref{eq:quadraticCost} and whose action can be computed by means of operations that respect the connectivity of the system. We consider that the computational operations required to obtain each component's control action $\vcu_{i}(t)$ from the state $\mtX(t)$ only involve information relied by other components to which component $i$ is connected. We call them \emph{distributed} controllers and denote them by
\begin{equation} \label{eq:decentralizedController}
    \mtU(t) = \fnPhi\big(\mtX(t); \stG\big)
\end{equation}
where $\fnPhi: \fdR^{N \times F} \to \fdR^{N \times G}$ and where the mapping is explicitly parametrized by the graph $\stG$. We have assumed, to ease the learning process, that the distributed controller is static.

The optimal distributed LQR controller can then be found by solving the following problem:
\begin{align}
    \min_{\fnPhi \in \fdPhi_{\stG}}\  & \fnJ \Big( \{\mtX(t)\}_{t=0}^{T}, \{\mtU(t)\}_{t=0}^{T-1} ; \mtQ, \mtR \Big) \label{eq:decentralizedLQRobj} \\
    \text{s. t. } & \mtU(t) = \fnPhi\big(\mtX(t); \stG \big) \label{eq:decentralizedLQRconstraint}
\end{align}
where $\fdPhi_{\stG}$ is the space of all possible mappings $\fdR^{N \times F} \to \fdR^{N \times G}$ that respect the structure of the graph.
Solving problem \eqref{eq:decentralizedLQRobj}-\eqref{eq:decentralizedLQRconstraint} requires solving an optimization problem over the field of functions $\fdPhi_{\stG}$. This is mathematically intractable in the general case \citep{Jahn07-NonlinearOptimization}, and requires specific approaches involving variational methods \citep{Cassel13-VariationalMethods}, dynamic programming \citep{Bertsekas05-DynamicProgramming} or kernel-based functions \citep[Ch. 14]{Murphy12-ProbabilisticML}. However, when we drop the distributed constraint \eqref{eq:decentralizedLQRconstraint}, problem \eqref{eq:decentralizedLQRobj} can actually be solved resulting in a linear controller; but this is a \emph{centralized} controller that does not respect the distributed nature of the system. In any case, methods for solving \eqref{eq:decentralizedLQRobj}-\eqref{eq:decentralizedLQRconstraint} require large datasets and exhibit poor generalization \citep[Ch. 6]{Goodfellow16-DeepLearning}.

Considering the inherent complexities of functional optimization, a popular approach is to adopt a specific model for the representation map $\fnPhi$ \citep[Ch. 2]{Anthony99-NeuralNetworkLearning}, leading to a parametric family of representations. Then, finding the best representation amounts to finding the optimal set of parameters, which results in a more tractable optimization problem over a finite-dimensional space \citep[Ch. 9]{Engl96-Inverse}. One such parametrization is that of distributed linear controllers $\fnPhi(\mtX(t);\stG) = \mtH(\stG) \mtX(t) \mtbH$, where we optimize over the space of all matrices $\mtH(\stG) \in \fdR^{N \times N}$ (and $\mtbH \in \fdR^{F \times G}$) that respect the connectivity pattern of the system. Many properties of this parametric family of controllers have been studied, including stability, robustness and (sub)optimality \citep{Fazelnia17-LowerBoundLQR, Fattahi19-LQR}.

However, it is known that even a linear system like \eqref{eq:linearSystem} may have a nonlinear optimal controller if we force a distributed nature on it \citep{Witsenhausen68-Counterexample}. This suggests that it would be more convenient to work with nonlinear parametrizations, rather than linear ones. In particular, we focus on graph neural networks (GNNs) \citep{Bruna14-DeepSpectralNetworks, Defferrard17-ChebNets, Gama19-Archit} as they are nonlinear mappings that exhibit several desirable properties. Fundamentally, they are naturally computed by means of local and distributed operations. This implies that any controller that is parametrized by means of a GNN respects the distributed nature of the system (as given by the graph $\stG$), naturally incorporating the distributed constraint into the chosen parametrization.


\section{Graph Neural Networks} \label{sec:GNN}



Graph signal processing (GSP) is a convenient framework to describe distributed problems \citep{Sandryhaila13-DSPG, Shuman13-SPG}. For a given graph $\stG = \{\stV, \stE\}$, we define a \emph{graph signal} as the mapping $\fnx: \stV \to \fdR^{F}$ that assigns an $F$-dimensional vector to each node, $\fnx(\lmv_{i}) = \vcx_{i} \in \fdR^{F}$. We can thus conveniently describe a graph signal by means of a matrix $\mtX \in \fdR^{N \times F}$ so that its $i^{\text{th}}$ row corresponds to the value of the signal at that node, $\fnx(\lmv_{i}) = \vcx_{i}$. It is immediate that we can consider both the system state $\mtX(t)$ and the control action $\mtU(t)$ as time-varying graph signals, $\fnx_{t}: \stV \to \fdR^{F}$ and $\fnu_{t}: \stV \to \fdR^{G}$, respectively [cf. \eqref{eq:stateMatrix}, \eqref{eq:linearSystem}].

Describing a graph signal by means of a matrix $\mtX$ is mathematically convenient but, in doing so, we rescind the relationship between the signal and the underlying graph support. In other words, the matrix $\mtX$ contains no information about the graph. To recover this relationship, we start by defining the \emph{support matrix} $\mtS \in \fdR^{N \times N}$, which respects the sparsity pattern of the graph, i.e. $[\mtS]_{ij} = 0$ for $i \neq j$ whenever $(\lmv_{j},\lmv_{i}) \notin \stE$. Therefore, the matrix $\mtS$ represents the underlying graph support, and examples in the literature include the adjacency \citep{Sandryhaila13-DSPG} and the Laplacian \citep{Shuman13-SPG} matrices, as well as their normalized counterparts \citep{Defferrard17-ChebNets}.

The support matrix $\mtS$ can then be used to define a linear mapping such that the output $\mtS \mtX$ is another graph signal whose values are related to the underlying graph support. More specifically, the $f^{\text{th}}$ output value at the $i^{\text{th}}$ component is given by
\begin{equation} \label{eq:graphShift}
    [\mtS \mtX]_{if} = \sum_{j=1}^{N} [\mtS]_{ij} [\mtX]_{jf} = \sum_{j : \lmv_{j} \in \stN_{i}} [\mtS]_{ij} [\mtX]_{jf}
\end{equation}
where $\stN_{i} = \{\lmv_{j}: (\lmv_{j},\lmv_{i}) \in \stE\} \cup \{\lmv_{i}\}$ is the neighborhood of node $\lmv_{i}$. Note that while the first equality corresponds to the definition of the matrix multiplication, the second equality holds because of the sparsity pattern of $\mtS$, i.e. the only nonzero entries in $\mtS$ correspond to those nodes that share an edge. In short, the linear map \eqref{eq:graphShift} yields an output graph signal that is a linear combination of neighboring values of the input graph signal.

The support matrix $\mtS$ acts as an elementary operator between graph signals. More precisely, it is a proper generalization of the unit time-shift (or time-delay) operator in traditional signal processing; it just shifts the signal through the graph (diffuses the signal), often receiving the name of \emph{graph shift operator} (GSO). Therefore, we can use $\mtS$ as the basic building block to construct linear graph filters as follows \citep{Sandryhaila13-DSPG}:
\begin{equation} \label{eq:graphFilter}
    \fnH(\mtX;\mtS) = \sum_{k=0}^{K} \mtS^{k} \mtX \mtH_{k}
\end{equation}
for some polynomial order $K$ and where $\mtH_{k} \in \fdR^{F \times G}$ is the corresponding \emph{filter tap} (or filter coefficient). Certainly, the output of the graph filtering operation \eqref{eq:graphFilter} is another graph signal, but with $G$ values at each node, so that $\fnH:\fdR^{N \times F} \to \fdR^{N \times G}$ is a linear map between graph signals.

Linear graph filters as in \eqref{eq:graphFilter} are local and distributed operations. To understand this, note that matrix multiplications to the left of $\mtX$ compute linear combinations of signal values across different nodes, while matrix multiplications to the right of $\mtX$ compute linear combinations of signal values within the same node. Therefore, for the operation to be local, multiplications on the right need to respect the sparsity of the graph, i.e. only combine values of nodes that are connected to each other (multiplications to the left can be arbitrary). This is precisely the case when multiplying by $\mtS^{k} = \mtS(\mtS^{k-1})$ since the operation can be computed by $k$ repeated exchanges with one-hop neighbors. Additionally, by storing the values of the filter taps $\{\mtH_{k}\}$, each node can compute the corresponding output separately by leveraging the information provided by the $k$ exchanges with its one-hop neighbors. Thus, a linear graph filter is a local and distributed operation, making them well suited for learning linear distributed controllers.

The graph filter \eqref{eq:graphFilter} is a proper generalization of the convolution operation and, thus, is usually referred to as a \emph{graph convolution} as well \citep{Gama20-GNNs}. Technically speaking, \eqref{eq:graphFilter} is a finite-impulse response (FIR) graph filter. However, due to the finite nature of graphs, it also encompasses infinite-impulse response (IIR) graph filters. In what follows, we focus on filters of the form \eqref{eq:graphFilter} and generically refer to them as graph filters.

Graph filters \eqref{eq:graphFilter} naturally model linear distributed controllers, i.e. $\mtU(t) = \fnH(\mtX;\mtS)$. However, we are interested in learning distributed controllers that are capable of capturing nonlinear relationships between the state and the control action. Towards this end, we introduce GNNs \citep{Gama20-GNNs}, which cascade $L$ layers, each applying a graph filter followed by a pointwise nonlinearity
\begin{equation} \label{eq:GNN}
    \fnPhi(\mtX; \mtS) = \mtX_{L} \quad , \quad \mtX_{\ell} = \sigma \Big( \fnH_{\ell}(\mtX_{\ell-1}; \mtS) \Big)
\end{equation}
for $\ell=1,\ldots,L$, with $\mtX_{0} = \mtX$ being the input graph signal. The function $\sigma: \fdR \to \fdR$ is a pointwise nonlinearity that acts on each entry of the graph signal (in a slight abuse of notation, we write $\sigma(\mtX)$ to denote $[\sigma(\mtX)]_{ij} = \sigma([\mtX]_{ij})$). The output of each layer is a graph signal $\mtX_{\ell} \in \fdR^{N \times F_{\ell}}$ with $F_{0} = F$ and $F_{L} = G$. The specific nonlinearity $\sigma$ to be used, the size of the graph signals $F_{\ell}$, and the number of filters taps $K_{\ell}$ are design choices. We remark that once these values are chosen, the space of possible GNN-based controllers is completely characterized by the values of the filters taps $\{\mtH_{\ell k}\}_{\ell,k}$ at each layer \citep{Gama20-GNNs}.

The GNN \eqref{eq:GNN} exhibits several desirable properties for learning distributed nonlinear controllers. Fundamentally, due to the pointwise nature of the nonlinearity, they retain the local and distributed nature of graph filters. This means that their output can be computed separately at each node, by exchanging information with one-hop neighbors only. Additionally, they are permutation equivariant and Lipschitz continuous to changes in the graph support $\mtS$ \citep{Gama20-Stability}. These two properties show how the GNNs effectively exploit the graph structure of the system to improve learning (and thus, they are expected to work better when the dynamics are also graph-dependent), and facilitate scalability and transferability. Finally, we note that while permutation equivariance and Lipschitz continuity to perturbations are properties also exhibited by linear graph filters, the GNNs leverage the nonlinearity to increase the discriminative power, helping to capture more information that graph filters do \citep{Pfrommer20-Discriminability}.

All GNN-based controllers are naturally distributed. However, we are interested in those that exhibit a small quadratic cost. That is, we are interested in finding the appropriate filter taps such that \eqref{eq:quadraticCost} is minimized
\begin{align}
\min_{\{\mtH_{\ell k}\}_{\ell, k}}\  & \fnJ \Big( \{\mtX(t)\}_{t=0}^{T}, \{\mtU(t)\}_{t=0}^{T-1} ; \mtQ, \mtR \Big) \label{eq:decentralizedLQRGNNobj} \\
\text{s. t. } & \mtU(t) = \fnPhi\big(\mtX(t); \mtS \big). \label{eq:decentralizedLQRGNNconstraint}
\end{align}
Note that problem \eqref{eq:decentralizedLQRGNNobj}-\eqref{eq:decentralizedLQRGNNconstraint} is a finite-dimensional optimization one that has $\sum_{\ell=1}^{L} F_{\ell} F_{\ell-1} K_{\ell}$ dimensions (independent of the size of the system $N$). Note, also, that the distributed constraint \eqref{eq:decentralizedLQRconstraint} has been incorporated by forcing $\fnPhi$ to be a GNN \eqref{eq:GNN} [cf. \eqref{eq:decentralizedLQRGNNconstraint}].

Problem \eqref{eq:decentralizedLQRGNNobj}-\eqref{eq:decentralizedLQRGNNconstraint} is nonconvex due to the GNN-based controller constraint \eqref{eq:decentralizedLQRGNNconstraint}. Thus, to approximately solve this problem, we leverage an empirical risk minimization (ERM) approach that is typical in learning \citep{Vapnik00-StatisticalLearning}. To do this, we create a \emph{training set} $\stT = \{\mtX_{1,0}, \ldots, \mtX_{|\stT|,0}\}$ containing $|\stT|$ samples $\mtX_{p,0}$ drawn independently from some distribution $\fnp$, which we consider to be different random initializations of the system. We then focus on the ERM problem given by
\begin{equation} \label{eq:decentralizedERM}
\begin{aligned}
\min_{\{\mtH_{\ell k}\}_{\ell, k}}\  & \sum_{p=1}^{|\stT|} \fnJ \Big( \{\mtX_{p}(t)\}_{t=0}^{T}, \{\mtU_{p}(t)\}_{t=0}^{T-1} ; \mtQ, \mtR \Big) \\
\text{s. t. } & \mtU_{p}(t) = \fnPhi\big(\mtX_{p}(t); \mtS \big) \\
    & \mtX_{p}(t+1) = \mtA \mtX_{p}(t) + \mtB \mtU_{p}(t) \mtbB\ , \ \mtX_{p}(0) = \mtX_{p,0}.
\end{aligned}
\end{equation}
Recall that $\mtA, \mtB, \mtQ, \mtR$ are given by the problem, and $\mtS$ is an appropriately chosen support matrix (i.e. adjacency, Laplacian, etc.). We solve \eqref{eq:decentralizedERM} by means of an algorithm based on stochastic gradient descent \citep{Kingma15-ADAM}, efficiently computing the gradient of $\fnJ$ with respect to the parameters $\mtH_{\ell k}$ by means of the back-propagation algorithm \citep{Rumelhart86-BackProp}. To estimate the performance of the \emph{learned} controllers (i.e. those obtained by solving \eqref{eq:decentralizedERM}), we generate a new set of initial states, called the test set, and compute the average quadratic cost \eqref{eq:quadraticCost} on the resulting trajectories. In essence, we transform the optimization problem \eqref{eq:decentralizedLQRGNNobj}-\eqref{eq:decentralizedLQRGNNconstraint} into a \emph{self-supervised} ERM problem \eqref{eq:decentralizedERM} that is solved through simulated data.

Henceforth, we focus on a two-layer GNN-based controller given by [cf. \eqref{eq:GNN}]
\begin{equation} \label{eq:GNNcontroller}
    \mtU(t) = \fnH_{2} \Big( \sigma \big( \fnH_{1} (\mtX(t); \mtS) \big); \mtS\Big).
\end{equation}
We do so because, in practice, each communication exchange between one-hop neighbors takes time. Thus, having many layers and/or large values of $K_{\ell}$ may require unrealistically fast communications or would only be applicable to particularly slow processes. For more details on time-varying graph signals and communication delays, please refer to \citet{Isufi19-Forecasting, Gama20-GNNs, Gama20-LearningControl}.

We now give a sufficient condition for a GNN-based learned controller to stabilize the system.
%
\begin{proposition} \label{prop:stability}
    Consider a linear dynamical system \eqref{eq:linearSystem} controlled by $\mtU(t) = \fnPhi(\mtX(t)) + \mtE(t)$ with $\fnPhi$ the GNN-based controller given by \eqref{eq:GNNcontroller} and with a nonlinearity $\sigma$ such that $|\sigma(x)| \leq |x|$ for all $x \in \fdR$; and where $\mtE(t)$ is a disturbance term or exploratory signal that satisfies $\sum_{t=0}^{\infty} \|\mtE(t)\| < \infty$. Then, the closed-loop system is input-state stable, i.e there exist constants $\beta_{0},\beta_{1} \geq 0$ such that
    \begin{equation} \label{eq:L2gain}
    \sum_{t=0}^{\infty} \| \mtX(t) \| \leq \beta_{0} + \beta_{1} \sum_{t=0}^{\infty} \| \mtE(t) \|
    \end{equation}
    as long as $\fnH_{1}$ and $\fnH_{2}$ satisfy
    \begin{equation} \label{eq:systStability}
    b c_{2}c_{1a} + c_{2}\scbc_{1b} < \big( 1-a \big) \big( 1 - c_{2} \scbc_{1b} \big)
    \end{equation}
    where
    \begin{equation} \label{eq:constants}
    a = \| \mtA \| \ , \   b = \| \mtB \| \ , \ c_{2} = \sum_{g=1}^{G} c_{2}^{g} \ , \  c_{1a} = \sum_{g=1}^{F_{1}} c_{1a}^{g} \ ,\
    \scbc_{1b}  = \sum_{g=1}^{F_{1}} \scbc_{1b}^{g}
    \end{equation}
    for $\|\cdot\|$ is the spectral norm, and with $\mtH_{\ell}^{fg}(\mtS) = \sum_{k=0}^{K_{\ell}} [\mtH_{\ell k}]_{fg} \mtS^{k}$ being a polynomial built with the $(f,g)$ entries of the filter taps $\{\mtH_{k}\}$, $\mtbH_{1k} = \mtbB \mtH_{1k}$ and
    \begin{equation}
        c_{2}^{g} = \max_{f=1,\ldots,F_{1}} \| \mtH_{2}^{fg}(\mtS) \| \ , \ c_{1a}^{g} = \max_{f=1,\ldots,F} \| \mtH_{1}^{fg} (\mtS) \mtA \| \,\
        \scbc_{1b}^{g} = \max_{f=1,\ldots,G} \| \mtbH_{1}^{fg} (\mtS) \mtB \|.
    \end{equation}
\end{proposition}
\begin{proof}
    See Supplementary Material.
\end{proof}
%
Proposition~\ref{prop:stability} is a sufficient condition for the closed-loop system to be input-state stable. We note that the condition on the nonlinearity is mild and is satisfied by most popular nonlinearities (ReLUs, $\tanh$, $\sigmoid$, etc.). We remark that \eqref{eq:systStability} is a conservative bound, in that systems that do not satisfy it may be input-state stable as well.


\section{Numerical Experiments} \label{sec:sims}



We showcase the performance of GNNs for learning decentralized controllers in the linear-quadratic problem. We consider the traditional finite-time horizon formulation,
\begin{equation} \label{eq:LQsingle}
\begin{aligned}
    \min\quad & \vcx(T)^{\Tr} \mtQ \vcx(T) + \sum_{t=0}^{T-1} \Big( \vcx(t)^{\Tr} \mtQ \vcx(t) + \vcu(t)^{\Tr} \mtR \vcu(t) \Big) \\
    \text{subject to } & \vcx(t+1) = \mtA \vcx(t) + \mtB \vcu(t) \ , \ t = 0,1,\ldots,T-1
\end{aligned}
\end{equation}
for some given matrices $\mtA, \mtB, \mtQ, \mtR \in \fdR^{N \times N}$ with $\mtQ , \mtR$ being positive definite. We note that $\vcx(t),\vcu(t) \in \fdR^{N}$ so that \eqref{eq:LQsingle} is a particular case of \eqref{eq:linearSystem}-\eqref{eq:quadraticCost} with $F=G=1$, where $\mtbB = \scbb$ has been absorbed in $\mtB$, and where we use the $\ell_{2}$ norm in \eqref{eq:quadraticCost}. In short, we consider the state and the control action of each node to be a scalar, $[\vcx(t)]_{i} = x_{i}(t)$ and $[\vcu(t)]_{i}=u_{i}(t)$, respectively.

\medskip\noindent \textbf{Setting.} We consider a system with $N$ nodes placed uniformly at random in a plane, and we build the corresponding graph $\stG$ by keeping only the $5$-nearest neighbors of each node. We adopt a support matrix $\mtS$ given by the adjacency matrix and normalized by the largest eigenvalue $\|\mtS\|_{2}=1$. The system matrices $\mtA$ and $\mtB$ are chosen at random and we set $\mtQ = \mtR = \mtI$. Unless otherwise specified, we set $N=20$ nodes, $T=50$ instances, and the matrices $\mtA$ and $\mtB$ are made to share the same eigenvector basis as $\mtS$, while the eigenvalues are chosen at random using a zero-mean unit-variance Gaussian, and are then normalized to have $\|\mtA\|_{2} = 0.995$ and $\|\mtB\|_{2} = 1$.

\medskip\noindent \textbf{Controllers.} We consider $5$ different controllers. (i: Optim) As a baseline, we use the optimal, linear and centralized controller, which can be computed recursively by knowing $\mtA,\mtB,\mtQ,\mtR$. (ii: MLP) We consider a \emph{learned} centralized controller given by a two-layer fully-connected neural network with $F_{\text{MLP}}N$ hidden units in the first layer, $N$ output units, and a $\fnsigma_{\text{MLP}}$ nonlinearity. (iii: D-MLP) The decentralized controller in \citep{Huang05-LargeScaleDecentralized} which assigns an individual two-layer fully connected neural network to each node with $\fnsigma_{\text{D-MLP}}$ nonlinearity. (iv: GF) A decentralized linear graph filter bank controller [cf. \eqref{eq:graphFilter}] consisting of a cascade of two banks with $F_{0}=F_{2}= 1, F_{1} = F_{\text{GF}}$ and $K_{1} = K_{\text{GF}}, K_{2}=0$. (v: GNN) A two-layer GNN [cf. \eqref{eq:GNN}, \eqref{eq:GNNcontroller}] with $F_{0}=F_{2} = 1, F_{1} = F_{\text{GF}}$, and $K_{1} = K_{\text{GF}}, K_{2}=0$, and a $\fnsigma_{\text{GNN}}$ nonlinearity. Unless otherwise specified, all nonlinearities are $\tanh$, and we set $F_{\text{MLP}} = F_{\text{GNN}} = 32$, $F_{\text{GF}} = 16$ and $K_{\text{GF}} = K_{\text{GNN}} = 3$.

\medskip\noindent \textbf{Training and evaluation.} For training, we consider the ERM equivalent problem of \eqref{eq:LQsingle} [cf. \eqref{eq:decentralizedLQRGNNobj}-\eqref{eq:decentralizedLQRGNNconstraint}, \eqref{eq:decentralizedERM}] which we minimize over a training set consisting of $|\stT|$ initial states $\vcx_{p,0}$, $p=1,\ldots,|\stT|$, with entries randomly distributed following a zero-mean unit variance Gaussian. We adopt an ADAM optimizer with learning rate $\mu$ and forgetting factors $0.9$ and $0.999$ \citep{Kingma15-ADAM}. The training procedure consists of $30$ epochs with a batch size of $20$. Moreover, every $5$ training steps we run an evaluation over a validation set consisting of $50$ samples. After the training is finished, we retain the model parameters that have resulted in the lowest validation cost. For evaluation, we create a test set of $50$ samples in analogous fashion and run the resulting controllers. For each controller, we compute the cost given by \eqref{eq:LQsingle} for each trajectory, and average over all $50$ trajectories. For a fair comparison, we normalize the cost by the lower bound for decentralized controllers provided in \citep{Fazelnia17-LowerBoundLQR}. To account for the randomness in the data generation, for each experiment we run $10$ different graph realizations, and we run $10$ different instances of matrices $\mtA,\mtB$ for each graph realization. The reported results include the average over these realizations as well as the estimated standard deviation. Unless otherwise specified, we set $|\stT|=500$ and $\mu=0.01$ (except for training MLP where we set $\mu = 0.001$).

\begin{table}[t]
    \scriptsize
    \centering
    \caption{Average (std. deviation) normalized cost for different hyperparameters. Cost of other controllers: (i: Optim) $0.9961(\pm 0.0007)$, (ii: MLP) $0.999 (\pm 0.002)$, (iii: D-MLP) $1.11 (\pm 0.02)$. }
    \begin{tabular}{c|ccc}
        $F_{\text{GF}} / K_{\text{GF}}$ & $1$ & $2$ & $3$  \\ \hline
        $16$ & $10    (\pm 10)$   & $1.5 (\pm 0.5)$   & $\mathbf{1.21 (\pm 0.05)}$ \\
        $32$ & $1.6  (\pm 0.8)$  & $1.3 (\pm 0.2)$   & $1.3  (\pm 0.2)$ \\
        $64$ & $1.5  (\pm 0.3)$  & $1.6 (\pm 1.0)$   & $2    (\pm 1)$
    \end{tabular}
    \hfill
    \begin{tabular}{c|ccc}
        $F_{\text{GNN}} / K_{\text{GNN}}$ & $1$ & $2$ & $3$  \\ \hline
        $16$ & $1.26 (\pm 0.07)$ & $1.20 (\pm 0.05)$ & $1.18 (\pm 0.04)$ \\
        $32$ & $1.26 (\pm 0.08)$ & $1.19 (\pm 0.05)$ & $\mathbf{1.17 (\pm 0.05)}$ \\
        $64$ & $1.26 (\pm 0.07)$ & $1.20 (\pm 0.05)$ & $1.18 (\pm 0.05)$
    \end{tabular}
    \label{tab:hparam}
\end{table}

\medskip\noindent \textbf{Experiment 1: Design hyperparameters.} In the first experiment, we study the dependency of (iv: GF) and (v: GNN) with the number of features $F_{\text{GF}},F_{\text{GNN}}$ and number of filter taps $K_{\text{GF}},K_{\text{GNN}}$. Results are shown in Table~\ref{tab:hparam}. First, it can be observed that for the GF controller, the behavior with $F_{\text{GF}}$ and $K_{\text{GF}}$ is erratic, as evidenced not only by the different average costs, but also by the larger standard deviations. Notice that considering information from farther away neighbors (increasing $K_{\text{GF}}$) improves performance. The performance of the GNN controller, on the other hand, is more consistent and exhibits relatively good results in all cases. This does not appear to be affected by the number of chosen features $F_{\text{GNN}}$, but it improves with increasing the neighborhood information $K_{\text{GNN}}$. We have also ran experiments for different values of learning rate $\mu \in \{0.001, 0.005, 0.01\}$, and observed that this hyperparameter considerably impacts the performance of the learned controllers. In particular, the GF controller fails for almost all values of $F_{\text{GF}}$ and $K_{\text{GF}}$ when $\mu \in \{0.001, 0.005\}$. The GNN controller, on the other hand, succeeds for other values of $\mu$ but at a slightly higer cost ($1.21 (\pm 0.05)$ for $\mu = 0.001$ and $1.18 (\pm0.05)$ for $\mu = 0.005$). From this experiment, we adopt the values of $F_{\text{GF}} = 16, F_{\text{GNN}} = 32, K_{\text{GF}} = K_{\text{GNN}} = 3$ and $\mu = 0.01$ to be used from now on.

\medskip\noindent \textbf{Experiment 2: Comparison between controllers.} In the above experiment, we also computed the optimal centralized controller (i: Optim), a learnable centralized controller (ii: MLP) and a learnable decentralized controller (iii: D-MLP). We ran the methods for different values of $\mu$ and $F_{\text{MLP}}$ and kept the ones with the best performance, namely $\mu = 0.001$ and $F_{\text{MLP}} = 32$ for MLP and $\mu = 0.01$ for D-MLP. The normalized costs obtained are $0.9961 (\pm 0.0007)$ for Optim, $0.999 (\pm 0.002)$ for MLP and $1.11(\pm 0.02)$ for D-MLP. We note that the centralized controllers perform better, as expected, and, in fact, they have a lower cost than the lower bound \citep{Fazelnia17-LowerBoundLQR}. This is expected since the space of centralized controllers contains the space of decentralized ones, and thus any optimal centralized controller is bound to be at least as good as any optimal decentralized one. Next, we observe that D-MLP has a lower cost than GF and GNN. This is also expected since the representation space of D-MLP contains that of GNNs. However, as we see in Experiment 4, this controller does not scale, since the optimization space grows with the size of the network, making it increasingly difficult to navigate during the training process.

\begin{figure}[t]
    \centering
    \includegraphics[width=0.3\textwidth]{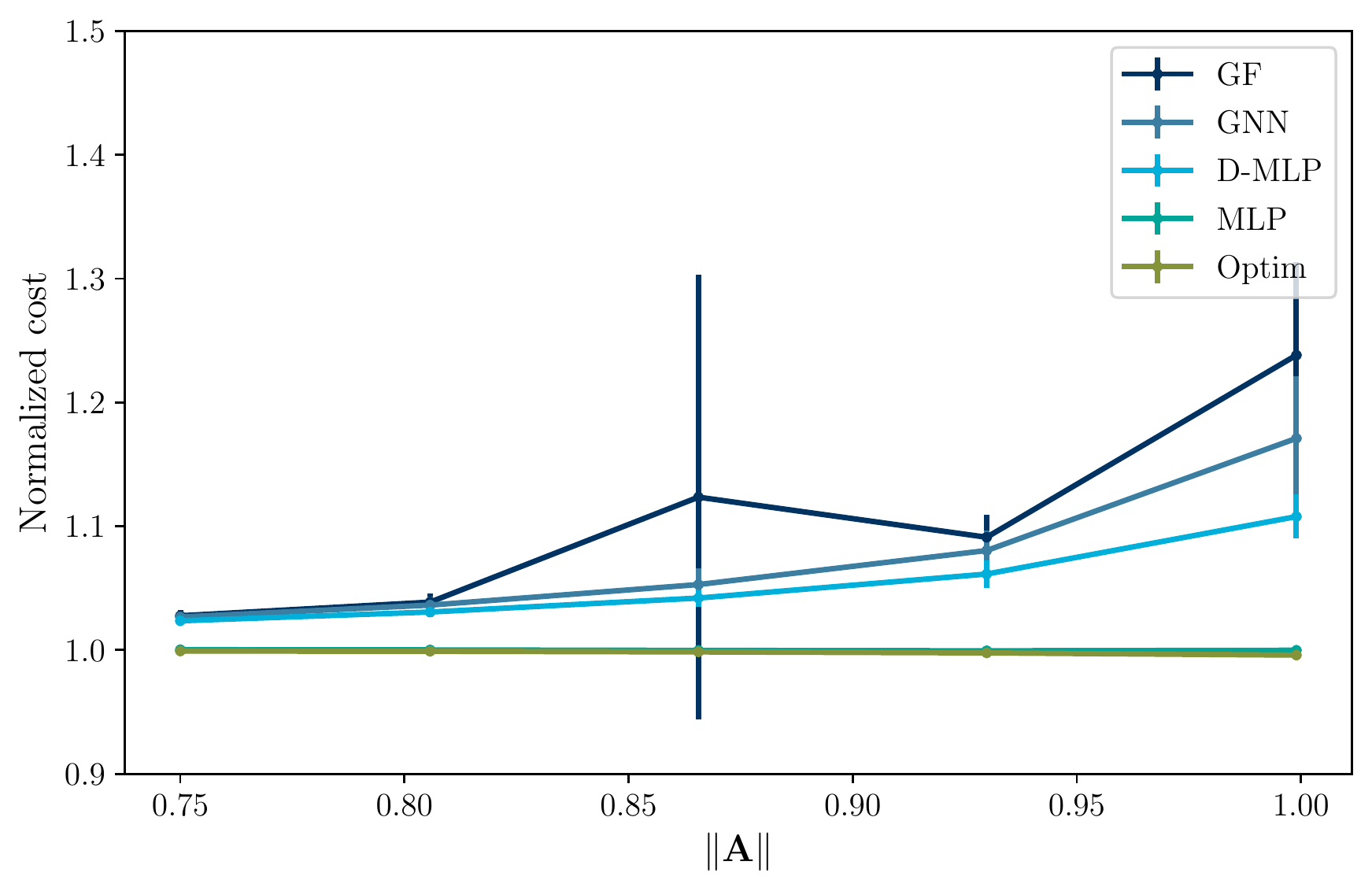}
    \hspace{0.5cm}
    \includegraphics[width=0.3\textwidth]{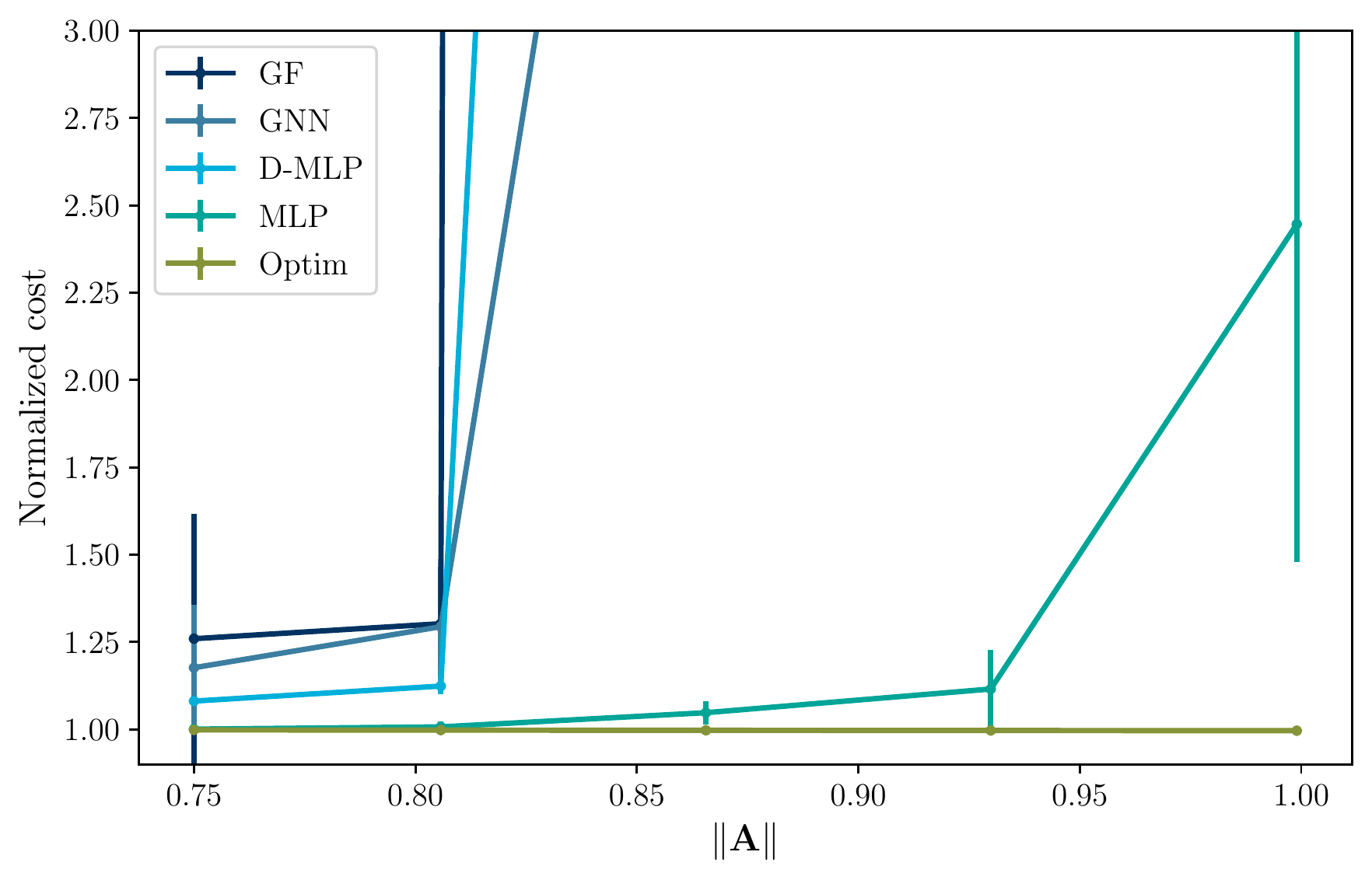}
    \hspace{0.5cm}
    \includegraphics[width=0.3\textwidth]{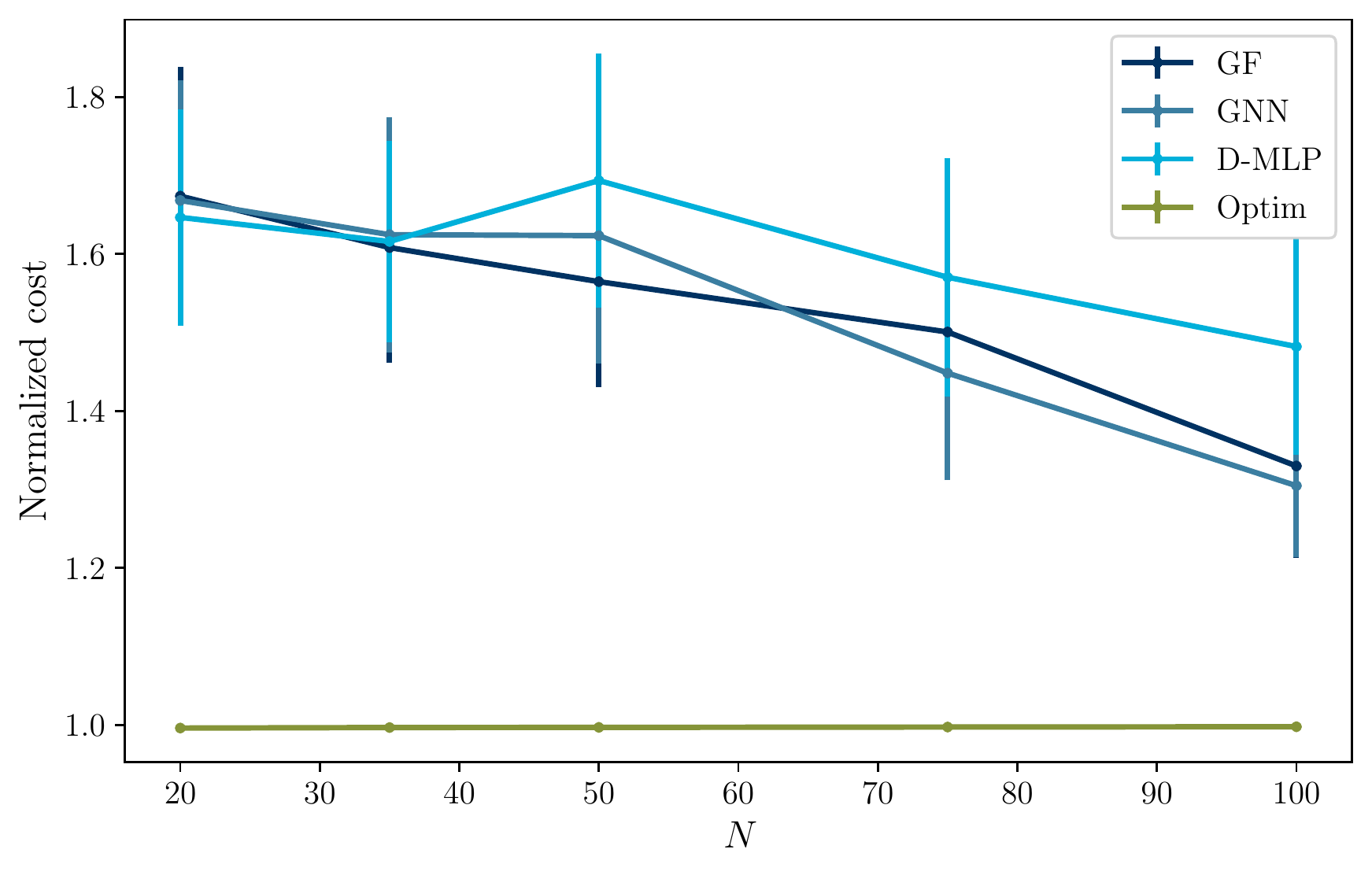}
    \caption{(Left) Change in cost for structured matrix $\mtA$. (Middle) Change in cost for random matrix $\mtA$. The higher the value of $\|\mtA\|$ the harder it is to control the system. We see that when the system matrix shares the same structure as the graph (i.e. same eigenvectors), the controllers are better at learning how to handle the system. We see that, when the structure is random, all decentralized controllers fail to learn appropriate actions when $\|\mtA\| > 0.8$. (Right) Change in cost for increasingly bigger systems. We observe that the D-MLP controller does not transfer at scale since, for bigger systems, it yields a higher cost than the GNN-based controller.}
    \label{fig:allResults}
\end{figure}

\medskip\noindent \textbf{Experiment 3: Dependence on the system matrix $\mtA$.} In this experiment, we analyze the ability of the learned controller to successfully handle different system matrices $\mtA$. We consider, first, the case when $\mtA$ and $\mtS$ have the same set of eigenvectors; and, second, when $\mtA$ is completely random and bears no relationship with the underlying graph support. Additionally, we consider different values of $\|\mtA\|$ where we recall that the larger the value of $\|\mtA\|$ is, the harder the system is to control. The results are shown in Fig.~\ref{fig:allResults}. First, it is evident that when the system dynamics share the same structure as the underlying graph support, the decentralized learning methods are capable of capturing this structure to improve their performance. Even as $\|\mtA\|$ gets closer to $1$ and the system becomes less stable [cf. Prop.~\ref{prop:stability}], the learned controllers work well. Second, when the system matrix $\mtA$ is completely arbitrary, then the decentralized controllers fail.

\medskip\noindent \textbf{Experiment 4: Transferability and scalability.} In the last experiment, we consider the case in which we train the learnable controllers in a graph with $N=20$ nodes, but then we test it on increasingly bigger systems of $N \in \{35, 50, 75, 100\}$ nodes. The results show in Fig.~\ref{fig:allResults} illustrate that, while the D-MLP performs better when tested in a small system, it does not transfer to larger systems. This is because it assigns a different fully connected neural network controller to each component. Thus, when tested on larger systems, it has to replicate this controller on other components and that may have a substantially different topological neighborhood. GNN-based controller, on the other hand, successfully adapt to larger systems, even when trained on small ones.


\section{Conclusions} \label{sec:conclusions}




Finding optimal distributed controllers is intractable in the most general case, and even in specific cases where the problem admits a convex formulation, the resulting controller does not scale up. Since it is expected that the optimal distributed controller is a nonlinear function of the states, we proposed to adopt a GNN to parametrize the controller. GNNs are well suited since they are nonlinear, naturally distributed, computationally efficient, scale and transfer. We observed in numerical experiments that they exhibit good performance.

This preliminary investigation of GNN-based controllers shows their potential for distributed control and opens up several avenues of future research. First, the sufficient condition for closed-loop stability is conservative and can be improved. Second, we can analyze the robustness of the GNN-based controller to changes in the system connectivity as well as in the matrix that describe the linear dynamics. Third, we can investigate the suboptimality of the GNN-based controllers.


\bibliography{bibFiles/myIEEEabrv,bibFiles/bibLQRGNN}


\clearpage
\pagenumbering{arabic}
\renewcommand*{\thepage}{A\arabic{page}}
\newpage

\begin{center}
    {\large \bf `Graph Neural Networks for Distributed Linear-Quadratic Control'\\Supplementary Materials}    
\end{center}

\appendix



\section{Proof of Proposition~\ref{prop:stability}}

For graph signals, we consider the norm given by $\| \mtX \|_{2,1} = \sum_{f=1}^{F} \| \vcx^{f}\|$, where $\vcx^{f} \in \fdR^{N}$ are the columns of $\mtX$ and $\|\cdot\|$ is the Euclidean vector norm. To prove Proposition~\ref{prop:stability}, consider the exploratory controller $\mtU(t) = \fnPhi(\mtX(t);\mtS)+\mtE(t)$ where $\fnPhi(\mtX(t);\mtS)$ is the two-layer GNN \eqref{eq:GNNcontroller} and $\mtE(t) \in \fdR^{N \times G}$ is an exploratory graph signal such that $\sum_{t=0}^{\infty} \|\mtE(t)\| < \infty$. The dynamics of the system then become
\begin{equation}
    \mtX(t) = \mtA \mtX(t-1) + \mtB \mtPhi(t-1) \mtbB+ \mtB \mtE(t-1) \mtbB
\end{equation}
where $\mtPhi(t-1) = \fnPhi(\mtX(t-1);\mtS)$ is the two-layer GNN \eqref{eq:GNNcontroller}.
To bound $\|\mtX(t)\|$, note that due to from the dynamics equation \eqref{eq:linearSystem}, $\mtX(t)$ depends on $\mtX(t-1)$, $\mtPhi(t-1)$ and $\mtE(t-1)$. Likewise, since the GNN-based controller \eqref{eq:GNNcontroller} can be seen as a \emph{(graph) recurrent neural network} ((G)RNN) \citep{Ruiz20-GRNN}
\begin{equation} \label{eq:GRNN}
\mtPhi(t) = \fnH_{2} \Big( \sigma \Big( \fnH_{1} \big(\mtA\mtX(t-1) \big) +  \fnH_{1} \big(\mtB\mtPhi(t-1)\mtbB \big) + \fnH_{1} \big(\mtB\mtE(t-1) \mtbB \big)\Big) \Big).
\end{equation}
which also depends on $\mtX(t-1)$, $\mtPhi(t-1)$ and $\mtE(t-1)$, and where we have dropped the explicit dependence of $\fnH_{\ell}$ on $\mtS$ for ease of exposition. This leads to a system of equations given by
\begin{align} \label{eq:systEqs}
& \mtX(t) = \mtA \mtX(t-1) + \mtB \mtPhi(t-1) \mtbB+ \mtB \mtE(t-1) \mtbB \\
& \mtU(t) = \fnH_{2} \Big( \sigma \Big( \fnH_{1} \big(\mtA\mtX(t-1) \big) +  \fnH_{1} \big(\mtB\mtPhi(t-1)\mtbB \big) + \fnH_{1} \big(\mtB\mtE(t-1) \mtbB \big)\Big) \Big) \nonumber
\end{align}
Now, we can apply norm of $\|\mtX(t)\|$ and bound it above by the triangular inequality and the submultiplicativity of the operator norm to arrive at
\begin{equation} \label{eq:ineqSignal}
\begin{aligned}
\| \mtX(t)\| &  \leq \| \mtA\| \sum_{f=1}^{F} \| \vcx^{f}(t-1)\| + \| \mtB \| \|\mtbB\|  \sum_{g=1}^{G} \| \vcphi^{g}(t-1)\| + \| \mtB \| \|\mtbB\|  \sum_{g=1}^{G} \| \vce^{g}(t-1)\| \\
\|\mtX(t)\|  & \leq \| \mtA \| \| \mtX(t-1)\| + \| \mtB \| \|\mtbB\| \| \mtPhi(t-1)\|+ \| \mtB \| \|\mtbB\| \| \mtE(t-1)\|
\end{aligned}
\end{equation}
which depends on $\|\mtX(t-1)\|$ recursively, but also on $\|\mtPhi(t-1)\|$ and $\|\mtE(t-1)\|$. We can then bound $\|\mtPhi(t)\|$ by applying the triangular inequality and leveraging the hypothesis $|\sigma(x)| \leq |x|$ for all $x \in \fdR$. First, we consider bounding filter $\fnH_{2}$ as follows:
\begin{equation} \label{eq:ineqControlH2init}
\| \mtU(t)\| = \sum_{g=1}^{G} \|\sum_{f=1}^{F_{1}} \mtH_{2}^{fg} (\mtS) \vcx_{1}^{f} \|  \leq \sum_{g=1}^{G} \sum_{f=1}^{F_{1}} \| \mtH_{2}^{fg}(\mtS)\| \| \vcx_{1}^{f} \| 
\end{equation}
where $\vcx_{1}^{f}$ is the $f^{\text{th}}$ column of the output of the first layer, $\mtX_{1} = \sigma(\fnH_{1}(\mtX(t);\mtS))$. Define $c_{2}^{g} = \max_{f=1,\ldots,F_{1}} \| \mtH_{2}^{fg}(\mtS)\|$ so that $\|\mtH_{2}^{fg}(\mtS) \| \leq c_{2}^{g}$ for all $f=1,\ldots,F_{1}$. Then, \eqref{eq:ineqControlH2init} becomes
\begin{equation} \label{eq:ineqControlH2}
\| \mtU(t)\|  \leq \sum_{g=1}^{G} c_{2}^{g}\sum_{f=1}^{F_{1}}  \| \vcx_{1}^{f} \|  = c_{2} \| \mtX_{1}\|
\end{equation}
where $c_{2} = \sum_{g=1}^{G} c_{2}^{g}$ as defined in \eqref{eq:constants}. Now, we proceed to bound
\begin{equation} \label{eq:X1}
\mtX_{1} = \sigma\Big( \fnH_{1} \big(\mtA \mtX(t-1) \big) + \fnH_{1}\big(\mtB \mtPhi(t-1) \mtbB\big) + \fnH_{1}\big(\mtB \mtE(t-1) \mtbB\big) \Big).
\end{equation}
Using the inequality $|\sigma(x)| \leq |x|$ together with the triangular inequality, we obtain
\begin{equation} \label{eq:X1sigma}
\|\mtX_{1}\| \leq \big\| \fnH_{1} \big(\mtA \mtX(t-1) \big) \big\|+\big\| \fnH_{1}\big(\mtB \mtPhi(t-1) \mtbB\big) \big\|+ \big\| \fnH_{1}\big(\mtB \mtE(t-1) \mtbB\big) \big\|.
\end{equation}
For the first term in \eqref{eq:X1sigma}, we have 
\begin{equation} \label{eq:X1firstInit}
\big\| \fnH_{1} \big(\mtA \mtX(t-1) \big) \big\|  = \sum_{g=1}^{F_{1}} \Big\| \sum_{f=1}^{F} \mtH_{1}^{fg}(\mtS) \mtA \vcx^{f}(t-1) \Big\| \leq \sum_{g=1}^{F_{1}} \sum_{f=1}^{F} \| \mtH_{1}^{fg}(\mtS) \mtA \| \|\vcx^{f}(t-1) \|
\end{equation}
and by defining $c_{1a}^{g}=\max_{f=1,\ldots,F} \| \mtH_{1}^{fg}(\mtS) \mtA\|$, we can write
\begin{equation} \label{eq:X1first}
\big\| \fnH_{1} \big(\mtA \mtX(t-1)\big) \big\| \leq \sum_{g=1}^{F_{1}} c_{1a}^{g} \sum_{f=1}^{F} \|\vcx^{f}(t-1) \| = c_{1a} \| \mtX(t-1)\|
\end{equation}
where $c_{1a}$ is defined in \eqref{eq:constants}. Now, for the second term, we first study the filtering stage
\begin{equation} \label{eq:X1secondFilter}
\fnH_{1} \big(\mtB \mtPhi(t-1) \mtbB\big) = \sum_{k=0}^{K_{1}} \mtS^{k} \mtB \mtPhi(t-1) \mtbB \mtH_{1k} = \sum_{k=0}^{K_{1}} \mtS^{k} \mtB \mtPhi(t-1) \mtbH_{1k} = \fnbH_{1} \big( \mtB \mtPhi(t-1)\big)
\end{equation}
with $\mtbH_{1k} = \mtbB \mtH_{1k}$, containing the filter taps of a modified filter $\fnbH_{1}: \fdR^{N \times G} \to \fdR^{N \times F_{1}}$. Then, we can proceed analogously to \eqref{eq:X1firstInit}-\eqref{eq:X1first} to obtain
\begin{equation} \label{eq:X1second}
\big\| \fnH_{1} \big(\mtB \mtPhi(t-1) \mtbB \big) \big\| \leq \sum_{g=1}^{F_{1}} \scbc_{1b}^{g} \sum_{f=1}^{G} \|\vcphi^{f}(t-1) \| = \scbc_{1b} \| \mtPhi(t-1)\|
\end{equation}
where we defined $\scbc_{1b}^{g} = \max_{f=1,\ldots,G} \| \mtbH_{1}^{fg}(\mtS) \mtB\|$ and then used the definition of $\scbc_{1b}$ given in \eqref{eq:constants}. Proceeding analogously to \eqref{eq:X1secondFilter} and \eqref{eq:X1second}, the third term of \eqref{eq:X1sigma} can be bounded as
\begin{equation} \label{eq:X1third}
\big\| \fnH_{1} \big(\mtB \mtE(t-1) \mtbB \big) \big\| \leq \sum_{g=1}^{F_{1}} \scbc_{1b}^{g} \sum_{f=1}^{G} \|\vce^{f}(t-1) \| = \scbc_{1b} \| \mtE(t-1)\|.
\end{equation}
Combining \eqref{eq:X1first}, \eqref{eq:X1second} and \eqref{eq:X1third}, we bound $\mtX_{1}$ \eqref{eq:X1} by
\begin{equation} \label{eq:X1bound}
\|\mtX_{1}\| \leq c_{1a} \| \mtX(t-1)\|+ \scbc_{1b} \|\mtPhi(t-1)\|+ \scbc_{1b} \|\mtE(t-1)\|.
\end{equation}
Using \eqref{eq:X1bound} in \eqref{eq:ineqControlH2}, we obtain
\begin{equation} \label{eq:ineqControl}
\| \mtU(t)\| \leq c_{2} \Big( c_{1a} \| \mtX(t-1)\| + \scbc_{1b} \| \mtU(t-1)\| + \scbc_{1b} \| \mtE(t-1)\| \Big).
\end{equation}

Considering \eqref{eq:ineqSignal} and \eqref{eq:ineqControl} jointly leads to a linear system of inequalities. To see this, define
\begin{equation} \label{eq:ineqDefs}
\begin{gathered}
x_{t} = \| \mtX(t)\| \quad , \quad \scphi_{t} = \|\mtPhi(t)\| \quad , \quad e_{t} = \|\mtE(t) \| \\ a =\|\mtA\| \quad , \quad b = \|\mtB\| \|\mtbB\| \quad , \quad \alpha = c_{2} c_{1a} \quad , \quad \beta = c_{2} \scbc_{1b}
\end{gathered}
\end{equation}
Note that $a,b$ are given by the problem, while $\alpha,\beta$ include the learnable parameters from the corresponding filters. With these definitions in place, we can write a system of linear inequalities as
\begin{equation} \label{eq:systIneqs}
\begin{bmatrix}
x_{t} \\
\scphi_{t}
\end{bmatrix}
\leq
\begin{bmatrix}
a & b \\
\alpha & \beta
\end{bmatrix}
\begin{bmatrix}
x_{t-1} \\
\scphi_{t-1}
\end{bmatrix}
+
\begin{bmatrix}
    b \\
    \beta
\end{bmatrix}
e_{t}
.
\end{equation}
By repeating this equation, we arrive at
\begin{equation} \label{eq:systIneqsInit}
\begin{bmatrix}
x_{t} \\
\scphi_{t}
\end{bmatrix}
\leq
\begin{bmatrix}
a & b \\
\alpha & \beta
\end{bmatrix}^{t}
\begin{bmatrix}
x_{0} \\
\scphi_{0}
\end{bmatrix}
+
\sum_{\tau = 0}^{t-1}
\begin{bmatrix}
    a & b \\
    \alpha & \beta
\end{bmatrix}^{\tau}
\begin{bmatrix}
    b \\
    \beta
\end{bmatrix}
e_{t-\tau-1}
\end{equation}
with $\scphi_{0} \leq c_{2} c_{1} x_{0}$, with $c_{1} = \sum_{g=1}^{F_{1}} \max_{f=1,\ldots,F} \| \mtH_{1}^{fg}(\mtS)\|$.

Assuming that the matrix in \eqref{eq:systIneqsInit} is diagonalizable, we can then rewrite
\begin{equation} \label{eq:matrixExp}
\begin{bmatrix}
a & b \\
\alpha & \beta
\end{bmatrix}^{t}
=
\begin{bmatrix}
v_{11} & v_{12} \\
v_{21} & v_{22}
\end{bmatrix}
\begin{bmatrix}
\lambda_{1}^{t} & 0 \\
0 & \lambda_{2}^{t}
\end{bmatrix}
\begin{bmatrix}
w_{11} & w_{12} \\
w_{21} & w_{22}
\end{bmatrix}
\end{equation}
where $v_{ij} = v_{ij}(a,b,\alpha,\beta)$ are the four elements of the eigenvector matrix, $\lambda_{i}=\lambda_{i}(a,b,\alpha,\beta)$ are the eigenvalues, and $w_{ij} = w_{ij}(a,b,\alpha,\beta)$ are the four elements of the inverse of the eigenvector matrix. This leads the first term in \eqref{eq:systIneqsInit} to
\begin{equation} \label{eq:systIneqsInitExp}
\begin{aligned}
& \begin{bmatrix}
x_{t} \\
\scphi_{t}
\end{bmatrix}
\leq
\begin{bmatrix}
v_{11} & v_{12} \\
v_{21} & v_{22}
\end{bmatrix}
\begin{bmatrix}
\lambda_{1}^{t} & 0 \\
0 & \lambda_{2}^{t}
\end{bmatrix}
\begin{bmatrix}
w_{11} & w_{12} \\
w_{21} & w_{22}
\end{bmatrix}
\begin{bmatrix}
x_{0} \\
\scphi_{0}
\end{bmatrix}
\\
& \begin{bmatrix}
x_{t} \\
\scphi_{t}
\end{bmatrix}
\leq
\begin{bmatrix}
v_{11} \lambda_{1}^{t}\big(w_{11}x_{0} + w_{12} \scphi_{0}\big)+ v_{12}\lambda_{2}^{t} \big(w_{21}x_{0} + w_{22} \scphi_{0} \big) \\
v_{21} \lambda_{1}^{t}\big(w_{11}x_{0} + w_{12} \scphi_{0}\big)+ v_{22}\lambda_{2}^{t} \big(w_{21}x_{0} + w_{22} \scphi_{0} \big)
\end{bmatrix}.
\end{aligned}
\end{equation}
From the first entry, we can find the bound on $x_{t}$ to be
\begin{equation} \label{eq:xtBoundControl}
x_{t} \leq \big(v_{11}w_{11} \lambda_{1}^{t} + v_{12}w_{21}\lambda_{2}^{t} \big) x_{0} + \big(v_{11}w_{12}  \lambda_{1}^{t} +v_{12}w_{22}\lambda_{2}^{t}\big) \scphi_{0}.
\end{equation}
Furthermore, the inequality $\scphi_{0} \leq c_{2}c_{1}  \scphi_{0}$ yields that
\begin{equation} \label{eq:xtBound}
x_{t} \leq \Big(v_{11}w_{11} \lambda_{1}^{t} + v_{12}w_{21}\lambda_{2}^{t}  + c_{2}c_{1} \big(v_{11}w_{12}  \lambda_{1}^{t} +v_{12}w_{22}\lambda_{2}^{t} \big) \Big) x_{0} 
\end{equation}
which can be conveniently regrouped as
\begin{equation} \label{eq:xtBoundTime}
x_{t} \leq \Big(\big(v_{11}w_{11} +c_{2}c_{1} v_{11}w_{12}\big) \lambda_{1}^{t} +\big( v_{12}w_{21}  + c_{2}c_{1} v_{12}w_{22} \big) \lambda_{2}^{t}\Big) x_{0}.
\end{equation}
Now, we can sum for all time instants as follows
\begin{equation} \label{eq:xtSumBound}
\sum_{t=1}^{\infty} x_{t} \leq  \Big( \big(v_{11}w_{11} + c_{2}c_{1} v_{11}w_{12}\big) \sum_{t=1}^{\infty}\lambda_{1}^{t} +\big( v_{12}w_{21} + c_{2}c_{1} v_{12}w_{22} \big)\sum_{t=1}^{\infty} \lambda_{2}^{t}\Big) x_{0}.
\end{equation}
For \eqref{eq:xtSumBound} to be finite, then $|\lambda_{1}| <1$ and $|\lambda_{2}|<1$. Leveraging that the eigenvalues of a $2 \times 2$ matrix are solutions to $\lambda^{2} - \trace[\mtA]\lambda + \det(\mtA) = 0$, we obtain
\begin{equation} \label{eq:charPol}
(a-\lambda)(\beta - \lambda)-b\alpha = \lambda^{2} - (a+\beta)\lambda+ (a\beta - b\alpha) = 0
\end{equation}
The roots of this polynomial are
\begin{equation} \label{eq:charPolRoots}
\lambda_{i} = \frac{(a+\beta) \pm \Big((a+\beta)^{2} - 4 (a \beta - b\alpha)\Big)^{1/2}}{2}
\end{equation}
which, noting that $(a+\beta)^{2} = a^{2} + 2a\beta+\beta^{2}$ and thus $(a+\beta)^{2}-4(a\beta-b\alpha) = a^{2} - 2a\beta+\beta^{2}+4b\alpha$, is equivalent to
\begin{equation} \label{eq:eigenvalues}
\lambda_{i} = \frac{(a+\beta) \pm \Big((a-\beta)^{2} + 4 b\alpha\Big)^{1/2}}{2}.
\end{equation}
For \eqref{eq:xtSumBound} to be finite the magnitude of both eigenvalues \eqref{eq:eigenvalues} has to be strictly smaller than $1$. By denote by $\gamma = ((a-\beta)^{2} + 4 b\alpha)^{1/2}$, this is equivalent to
\begin{equation} \label{eq:eigenvalueConditions}
-2 \leq a+\beta+\gamma < 2 \quad , \quad -2 \leq a+\beta-\gamma < 2
\end{equation}
which is, in turn, equivalent to
\begin{equation} \label{eq:eigenvalueConditionsGamma}
-2 - \gamma \leq a+\beta < 2 - \gamma \quad , \quad -2+\gamma \leq a+\beta < 2 + \gamma.
\end{equation}
Note that $a = \|\mtA\| \geq 0$ and $\beta = c_{2} \scbc_{1b} \geq 0$, and so is $\gamma \geq 0$. Thus, it always holds that $a+\beta \geq -(2+\gamma)$. Likewise, observe that either $2-\gamma > 0$ or $a+\beta = 0$ (which can only happen if both $a = \beta = 0$). Thus, $2-\gamma < 2+\gamma$. Therefore, the conditions in \eqref{eq:eigenvalueConditionsGamma} become
\begin{equation} \label{eq:eigenvalueConditionsSimple}
a+\beta < 2 - \gamma \quad , \quad \gamma < 2.
\end{equation}
Now, note that the first condition is $a+\beta+\gamma < 2$. However, since $a\geq 0$, $\beta \geq 0$ and $\gamma \geq 0$, this condition already implies that $\gamma < 2$. Thus, the conditions for convergence of \eqref{eq:xtSumBound} can be reduced to
\begin{equation} \label{eq:stabilityCondition}
a+\beta+\gamma < 2.
\end{equation}
Recalling that $\gamma = ((a-\beta)^{2} + 4 b\alpha)^{1/2}$, we can write
\begin{equation} \label{eq:stabilityConditionSqrt}
((a-\beta)^{2} + 4 b\alpha)^{1/2} < 2 - a - \beta.
\end{equation}
Noting that the left hand side of \eqref{eq:stabilityConditionSqrt}, we can square both terms to obtain
\begin{equation}
(a-\beta)^{2} + 4 b\alpha < (2 - (a + \beta))^{2}.
\end{equation}
Expanding both sides yields
\begin{equation}
a^{2}-2a\beta+\beta^{2} + 4 b\alpha < 4 - 4 a-4\beta + a^{2}+2a\beta+\beta^{2}
\end{equation}
which is equivalent to
\begin{equation} \label{eq:stabilityConditionSimplified}
a+\beta-a\beta +  b\alpha< 1.
\end{equation}

For the second term in \eqref{eq:systIneqsInit}, leverage \eqref{eq:matrixExp} to get the following bound for the first element in the vector
\begin{equation}
    \sum_{\tau=0}^{t-1} \scxi_{\tau} e_{t-\tau-1} \quad , \quad \scxi_{\tau} = \Big( \big(bv_{11}w_{11} + \beta v_{11}w_{12}  \big) \lambda_{1}^{\tau}  + \big(b v_{12}w_{21} +\beta v_{12}w_{22}\big) \lambda_{2}^{\tau} \Big).
\end{equation}
Now, add up over all time instants to get
\begin{equation} \label{eq:ineqExploratory}
 \sum_{t=0}^{\infty} \sum_{\tau=0}^{t-1} \scxi_{\tau} e_{t-\tau-1} \leq \Big(\sum_{t=0}^{\infty} \scxi_{t} \Big) \Big( \sum_{t=0}^{\infty} e_{t} \Big)
\end{equation}
where the inequality holds as long as both series on the right-hand side are finite. By hypothesis it holds that $\sum_{t=0}^{\infty} e_{t} < \infty$, so the focus is set on
\begin{equation} \label{eq:seriesXi}
\sum_{t=0}^{\infty} \scxi_{t} = \big(bv_{11}w_{11} + \beta v_{11}w_{12}  \big) \sum_{t=0}^{\infty} \lambda_{1}^{t}  + \big(b v_{12}w_{21} +\beta v_{12}w_{22}\big) \sum_{t=0}^{\infty} \lambda_{2}^{t}.
\end{equation}
If the condition \eqref{eq:stabilityConditionSimplified} is satisfied, then it holds that \eqref{eq:seriesXi} is finite, and thus \eqref{eq:ineqExploratory} holds. Therefore, combining \eqref{eq:xtSumBound} and \eqref{eq:ineqExploratory} we prove that the system is input-state stable as per \eqref{eq:L2gain}.

Finally, let us rewrite $\alpha, \beta$ in terms of quantities that depend on the GNN-based controller as follows. Note that \eqref{eq:stabilityConditionSimplified} can be rewritten as
\begin{equation}
-(1-a)\beta + (1-a)>  b\alpha
\end{equation}
with $\beta$ and $\alpha$ containing the learnable parameters (through filters $\fnH_{1}$ and $\fnH_{2}$). Finally, replace $\alpha = c_{2} c_{1a}$ and $\beta=c_{2} \scbc_{1b}$ to obtain
\begin{equation} \label{eq:stabilityConditionNoDivision}
bc_{2} c_{1a} < (1-a) -(1-a)c_{2}\scbc_{1b} = \big( 1-a \big) \big( 1 - c_{2} \scbc_{1b} \big)
\end{equation}
from which the result \eqref{eq:systStability} of Proposition~\ref{prop:stability} follows. \hfill $\blacksquare$

\end{document}